\documentclass[12pt,english]{article}
\usepackage[T2A,T2A,T1]{fontenc}
\usepackage[koi8-r]{inputenc}
\usepackage{textcomp}
\usepackage{amsmath}
\usepackage{amssymb}
\usepackage{esint}
\usepackage{amsthm}

\makeatletter

\newtheorem{theorem}{Theorem}[section]

\newtheorem{lemma}{Lemma}[section]
\newtheorem{proposition}{Proposition}[section]

\def\tire{\thinspace--\thinspace}
\let\epsilon=\varepsilon
\let\phi=\varphi
\title{Convergence to equilibrium due to collisions with external particles}
\author{A.A.~Lykov \thanks{Mechanics and Mathematics Faculty, Lomonosov Moscow
State University, Leninskie Gory~1, Moscow, 119991, Russia. E-mail:~alekslyk@yandex.ru, 2malyshev@mail.ru}, V.A.~Malyshev\footnotemark[1]
}

\makeatother

\begin{document}
\maketitle


\begin{abstract}
  We consider a class of ``most non ergodic'' particle systems and prove
  that for most of them ergodicity appears if only one particle of $N$ has
  contact with  external world, that is this particle  collides with
  external particles in random time moments.
\end{abstract}


\tableofcontents{}

\section{Introduction}

In this paper we prove the main results of our previous paper \cite{LM_4}
under more realistic assumptions. In \cite{LM_4} the external force
was assumed to consist of 2 terms -- continuous time random gaussian
process (white noise) and simplest deterministic dissipative force
$-\alpha v$. Now we consider the collisions of one internal particle
with external particles having random velocities at random time moments.
We show that the results appear to be quite similar to \cite{LM_4}.

\section{Model and results}

\subsection{Hamiltonian dynamics}

We consider system of $N$ point particles in ${\bf R}^{d}$ with
the linear phase space 
\[
L=\mathbb{{\bf R}}^{2dN} \! =\! \Bigl\{\psi=\left(\begin{array}{c}
q\\
p
\end{array}\right):\, q=\! (q_{1},\ldots ,q_{N})^{T},\, p=\! (p_{1},\ldots ,p_{N})^{T},\ p_{k},q_{k}\in{\bf R'}^{d}\Bigr\},
\]
where $()^T$ means transposition. Coordinates and momenta $q_{k}=(q_{k1},\ldots ,q_{kd})^T$, $p_{k}=(p_{k1},\ldots ,p_{kd})^T,$
are also considered as column vectors, so $\psi$ is a column vector.
$L$ has obvious direct sum representation 
\[
L=l_{2}^{(dN)}\oplus l_{2}^{(dN)},\; l_{2}^{(dN)}={\bf R}^{dN}.
\]
Introduce the Hamiltonian 
\begin{equation}
H(\psi)=\sum_{k=1}^{N}\frac{|p_{k}|^{2}}{2M}+\frac{1}{2}(q,Vq),\label{Ham1}
\end{equation}
where $M>0$, and $V$ is a positive definite $(dN\times dN)$-matrix.

The evolution is defined by the following Hamiltonian system of equations
on $L$: 
\begin{align}
\dot{q}_{kj}= & \frac{\partial H}{\partial p_{kj}},\label{L1}\\
\dot{p}_{kj}= & -\frac{\partial H}{\partial q_{kj}},\label{L2}
\end{align}
where $k=1,\ldots,N,j=1,\ldots ,d$. The system (\ref{L1}), (\ref{L2})
can be rewritten as
\begin{equation}
\dot{\psi}=A\psi,\quad \,\,\,A=\left(\begin{array}{cc}
0 & \tfrac{1}{M}E\\
-V & 0
\end{array}\right)\label{LHamVec}
\end{equation}
The solution $\psi(t)$ of (\ref{LHamVec}) with initial conditions
$\psi(0)$ is 
\[
\psi(t)=e^{tA}\psi(0).
\]

\subsection{Collisions}

Fix some particle, say with number $n=1$ and assume that at random
time moments 
\[
0=t_{0}<t_{1}<\ldots<t_{k}<\ldots
\]
the velocity of this particle changes instantaneously somehow, that
is there are jumps $p_{1}(t_{k}-)\to p_{1}(t_{k})$. On the intervals
$[t_{k-1},t_{k})$ the dynamics is defined by the equations (\ref{L1}),
(\ref{L2}) with the corresponding initial conditions $\psi(t_{k-1})$.
We interpret the nature of these jumps as collisions with external
particles on short time scale.

We use three exact models of ``collisions'': concrete 1-1, 1-$d$,
and abstract, 2.

\paragraph{Condition 1-1}

Assume that masses of external particles are equal to $m$ and internal
masses are equal to $M>m$ (this condition is not necessary for our
goals, but here it is technically convenient). We also assume that
$d=1$ and the collision conserves energy and momentum for this two
particle system. Then it is known that (see for example \cite{M_Intro_1})
\begin{equation}
v(t_{k})=\alpha v(t_{k}-)+(1-\alpha)u_{k}, \; 0<\alpha=\frac{M-m}{M+m}<1\label{cond_1-1}
\end{equation}
or 
\[
p(t_{k})=\alpha p(t_{k}-)+(1-\alpha)Mu_{k}
\]
where $v(t)$ is the velocity of the particle 1 at time $t$, and
$u_{k}$ is the velocity of external particle with which occurs the
collision at time $t_{k}$.

\paragraph{Condition 1-d}

For dimension $d>1$ we assume transformation 
\begin{equation}
v(t_{k})=Rv(t_{k}-)+w_{k},\label{cond_1-2}
\end{equation}
with some matrix (possibly random) $R$ satisfying the following conditions 
\begin{enumerate}
\item the distribution of the vector $w_{k}$ has everywhere positive density
on $\mathbb{{\bf R}}^{d}$; 
\item the matrix $R$ defines a contraction map of $\mathbb{{\bf R}}^{d}$,
i.e.\ there exists constant $0<\alpha<1$ such that 
\[
(Rp,Rp)\leqslant\alpha(p,p)
\]
for all $p\in\mathbb{{\bf R}}^{d}$. Or equivalently, the spectrum
of the matrix $RR^{T}$ lies in the open unit interval. 
\end{enumerate}

\paragraph{Condition 2}

These are more general but seemingly more technical conditions: 
\begin{equation}
p_{1}(t_{k})=J(\xi_{k},p_{1}(t_{k}-))\label{cond_2}
\end{equation}
where $\xi_{k}$ is a random $l$-dimensional vector $l\geqslant1$,
and 
\[
J:{\bf R}^{l}\times{\bf R}^{d}\to\mathbb{{\bf R}}^{d}=\{p_{1}\}
\]
is some vector function. For example, in case 1-1 the vector $\xi_{k}=u_{k}$.

Transformation $J$ is assumed to have the following properties: 
\begin{enumerate}
\item $J$ is everywhere differentiable and analytic almost (w.r.t.\ Lebesgue
measure) for any point $(g,p)\in\mathcal{O}_{\xi}\times{\bf R}^{d}$,
where $\mathcal{O}_{\xi}$ is some open set in ${\bf R}^{l}$. For
example, in case 1-1 $\mathcal{O}_{\xi}={\bf R}^{d}$. 
\item For any $p\in{\bf R}^{d}$ the image $J(\mathcal{O}_{\xi},p)\subset\mathbb{{\bf R}}^{d}$
contains sphere of radius $|p|$, and for any $h\geqslant0$ there
exists $g=g(p,h)\in\mathcal{O}_{\xi}$ such that 
\[
|J(g;p)|=h.
\]
\item There exists compact subset $K\subset L$ and positive constant $\alpha<1$
such that for any $p\notin K$ the following inequality holds: 
\[
E_{\xi}|J(\xi;p))|^{2}\leqslant\alpha|p|^{2},
\]
Intuitively this means that kinetic energy of particle 1, being sufficiently
large, should be decreased by collisions. 
\end{enumerate}

\paragraph{Assumptions on the nature of randomness}

We assume that: 
\begin{enumerate}
\item $\tau_{m}=t_{m}-t_{m-1},m=1,2,\ldots ,$ are i.i.d.\ random variables with
positive on $R_{+}$ density $\rho_{\tau}$ such that $E\tau_{1}<\infty$;
\item random vectors $\xi_{1},\xi_{2},\ldots,\xi_{m},\ldots$ are i.i.d.\ 
with distribution function $F_{\xi}$. Assume that $F_{\xi}$ has
density $\rho_{\xi}$ with respect to Lebesgue measure on $\mathbb{{\bf R}}^{l}$,
which is positive on some open set $\mathcal{O}_{\xi}\subset\mathbb{{\bf R}}^{l}$
and such that $P(\xi_{1}\in\mathcal{O}_{\xi})=1$; 
\item the arrays $(\xi_{1},\xi_{2},\ldots,\xi_{m},\ldots)$ and $(\tau_{1},\tau_{2}\ldots,\tau_{m},\ldots)$
are mutually independent. 
\end{enumerate}

\paragraph{Dynamics with collisions}

For $\psi=(q,p)\in L,\ q=(q_{1},\ldots,q_{N})^{T},\ p=(p_{1},\ldots,p_{N})^{T}$
and $\xi\in\mathbb{{\bf R}}^{l}$ we define the transformation: 
\[
J_{L}(\xi;\psi)=\psi'=(q',p'),
\]
where 
\[
q'=q,\quad p_{1}'=J(\xi,p_{1}),\ p_{k}'=p_{k},\ k>1.
\]
For simplicity we will omit index $L$ and denote $J_{L}(\cdot,\cdot)$
simply $J(\cdot,\cdot)$. For any $t\geqslant0$ we will consider
the following transformations of $L$, 
\[
J(t,\xi;\psi)=J(\xi;e^{tA}\psi),\ \psi\in L.
\]
For any $\psi\in L$ and any integer $m\geqslant1$ define the time
lengths $\tau_{1}=t_{1},\tau_{2}=t_{2}-t_{1},\ldots,\tau_{m}=t_{m}-t_{m-1}$.
For any $t\geq0$ we define the dynamics with collisions as follows
\[
\psi(t)=e^{A(t-t_{m})}\psi_{m},\ \psi_{m}=J(\tau_{m},\xi_{m};\psi_{m-1}),\ \psi_{0}=\psi(0).
\]
where $m$ is the maximal integer such that $t_{m}<t$.

Note that if $\tau_{k}$ are exponentially distributed then $\psi(t)$
is a Markov process. Moreover, the defined random process $\psi(t)$
is piecewise-deterministic continuous time Markov process with trajectories
continuous from the right.

\subsection{Dissipative subspace}

Denote by $\mathbf{H}$ the set of all positive definite $dN\times dN$
matrices $V$ . Note that the set of all symmetric matrices is the
linear space of dimension $\frac{dN(dN+1)}{2}$, and $\mathbf{H}$
is its open subset with induced topology and induced Lebesgue measure
$\lambda$ from the space of symmetric matrices.

Define the dissipative subspace 
\[
L_{-}(V)=L_{-}=\left\{\left(\begin{array}{c}
q\\
p
\end{array}\right)\in L:q,p\in l_{V}\right\}
\]
where $l_{V}$ is the subspace of ${\bf R}^{dN}$ generated by the
vectors $V^{k}e_{1},\ k=0,1,2\ldots$, where $e_{n},n=1,\ldots ,dN$ is
the standard basis of $l^{(dN)}={\bf R}^{dN}$.

We say that $V$ has\textbf{ completeness} property if 
\begin{equation}
L_{-}(V)=L. \label{completeness}
\end{equation}
Further we denote $\omega_{1}^{2},\ldots,\omega_{dN}^{2}$ all (positive)
eigenvalues of $V$.

Denote by $\mathbf{H}^{+}$ the set of positive definite $dN\times dN$
matrices $V$ having completeness property. Denote also $\mathbf{H}^{++}\subset\mathbf{H}^{+}$
the subset of matrices with completeness property and, moreover, having
spectrum such that $\omega_{1},\ldots \omega_{dN}$ are rationally independent.
Further on we assume all $\omega_{k}$ to be positive.

\begin{proposition}\label{prop_dissip}
The set $\mathbf{H}^{++}$ is everywhere dense in the set $\mathbf{H}$
of all positive definite $dN\times dN$ matrices, and moreover, the
complement $\mathbf{H}\smallsetminus\mathbf{H}^{++}$ has Lebesgue
measure zero.
\end{proposition}

\subsection{Main results}

In the theorems below the above mentioned assumptions are always assumed.

\paragraph{Ergodicity}

\begin{theorem} \label{thConvProcess} There exists probability measure
$\pi$ on $L$ such that:

1) it is absolutely continuous with respect to Lebesgue measure $\lambda$
on $L$ and has positive density;

2) for any measurable bounded function $f$ on $L$ and any initial
condition $\psi(0)$ 
\[
\lim_{T\to\infty}\frac{1}{T}\int_{0}^{T}f(\psi(t))dt=\int_{L}f(\psi)\pi(d\psi),\quad\mbox{a.s.}
\]
\end{theorem}

\paragraph{Stronger convergence of the embedded chain}

Denote by $\psi_{n}$ the embedded chain: 
\begin{equation}
\psi_{n}=\psi(t_{n}).\label{initChain}
\end{equation}
It is clear that $\psi_{n}$ is a discrete time Markov chain. Let
$P(\psi,A),\psi\in L,\ A\subset\mathcal{B}(L),$ be the transition
probability of this chain, where $\mathcal{B}(L)$ is the Borel $\sigma$-algebra
in $L$.

\begin{theorem} \label{thConvChain} $\psi_{n}$ has  unique invariant
measure $\pi$ on $L$ (up to multiplication on a constant). Moreover,
it has the following properties:
\begin{enumerate}
\item $\pi$ is the same in the previous theorem. Further on, we will suppose,
that $\pi(X)=1$.
\item For any $\psi\in L$ we have: 
$
\sup_{A\in\mathcal{B}(L)}\left|P^{k}(\psi,A)-\pi(A)\right|\rightarrow0,\ \mbox{as}\ k\rightarrow\infty.
$
\end{enumerate}
\end{theorem}

\paragraph{Properties of the invariant measure}

Here we give some results concerning invariant measure $\pi$ in case
 1-1. Assume moreover that $u_{m}$ have finite second moment $Eu_{m}^{2}<\infty$.

Assume also that the intervals $\tau_{m}$ between collisions have
exponential distribution with parameter $\lambda>0$. In this case
$\psi(t)$ is time homogeneous Markov process. Denote 
\[
C(t)=E_{\psi}\{(\psi(t)-E\psi(t))(\psi(t)-E\psi(t))^{T}\}
\]
its covariance matrix.

\begin{theorem} \label{limitMeasurePropertyTh} The following propositions
holds: 
\begin{enumerate}
\item The measure $\pi$ is Gibbs measure, that is, it has density 
\begin{align*}
  p_{\beta}(\psi)&=\frac{1}{Z}\exp(-\beta H(\psi),\\
  Z&=\int_{L}\exp(-\beta H(\psi)d\lambda(\psi)
\end{align*}
for some $\beta>0$ iff $u_{m}$ have gaussian distribution with mean
zero and some variance $\sigma^{2}$. In the gaussian case we have
\begin{equation}
\beta=\frac{(1+\alpha)}{M(1-\alpha)\sigma^{2}}=\frac{1}{m\sigma^{2}},\label{betaDef}
\end{equation}
where $m$ is the mass of external particles.
\item However, if $Eu_{m}=0$ but $u_{m}$ is not gaussian, then for
any $\psi(0)$ the limiting covariance matrix is the same as above,
i.e. 
\[
\lim_{t\rightarrow\infty}C(t)=\beta^{-1}\left(\begin{array}{cc}
V^{-1} & 0\\
0 & ME
\end{array}\right),
\]
where $\beta$ is defined in (\ref{betaDef}). 
\end{enumerate}
\end{theorem}

\paragraph{Plan of the proof}

First of all we will prove Theorem \ref{thConvChain} about convergence
of the embedded chain. For this we will need some results from the
theory of general Markov processes, which will be formulated just
now. We will show in Appendix how they follow from the known results.
After this, we will prove Theorems \ref{thConvProcess} and \ref{limitMeasurePropertyTh}.

\paragraph{Stability of Markov processes}

Let $X$ be complete locally compact separable metric space and $\mathcal{B}(X)$
-- its Borel $\sigma$-algebra. Also we will consider measures on $X$
which will be always assumed non-negative and countably additive.
Consider Markov chain $\xi_{n},\ n=0,1,2,\ldots$ on $X$ with transition
probabilities $P(x,A),\ x\in X,A\in\mathcal{B}(X)$.

\begin{theorem} \label{chainConvThCommonX} Assume that the following
conditions hold: 
\begin{enumerate}
\item (strong irreducibility property) there exist integer $m\geqslant1$
such that for any $x\in X$ the m-step transition probability $P^{m}(x,\cdot)$
is equivalent to some finite measure $\mu$, having the property that
$\mu(O)>0$ for any open subset $O\subset X$;
\item (weak Feller property) for any open $O\subset X$ the function $P(x,O)$
is lower semi-continuous;
\item (drift condition) there exists compact subset $K\subset X$ and non-negative
measurable function $f(x),x\in X$ (Lyapunov function), which tends
to infinity with the distance from K, such that 
\[
\int f(y)dP(x,dy)-f(x)\leqslant-1,\ \mbox{for all}\ x\in X\setminus K,
\]
\[
\int f(y)dP(x,dy)<\infty\ \mbox{for all}\ x\in K.
\]
\end{enumerate}
Then there exists a unique (up to a multiplicative constant) invariant
measure $\pi$ for $\xi_{n}$ and the following properties hold: 
\begin{enumerate}
\item $\pi$ is finite and absolutely continuous w.r.t.\ measure $\mu$.
We assume further that $\pi(X)=1$.
\item For any $x\in X$ 
\[
\lim_{n\rightarrow\infty}\sup_{A\subset\mathcal{B}(X)}|P^{n}(x,A)-\pi(A)|=0.
\]
\item For any bounded measurable function $h$ on $X$ and any initial condition
$\xi_{0}$ we have: 
\[
\frac{1}{n}\sum_{k=0}^{n}h(\xi_{k})\rightarrow\pi(h)=\int_{X}h(x)\pi(dx),\ \mbox{as}\ n\rightarrow\infty,\ \mbox{a.s.}
\]
\end{enumerate}
\end{theorem}

\section{Proof of Theorem \ref{thConvChain} }

Consider now the embedded Markov chain 
$
\psi_{k}=\psi(t_{k}),\quad\psi_{0}=\psi(0).
$
For $\psi\in L$ and Borel subset $A\subset L$ let $P(\psi,A)$ be
the transition probability of the Markov chain $\psi_{k}$. Without
loss of generality we assume here $M=1$.

Theorem \ref{thConvChain} immediately follows from theorem
\ref{chainConvThCommonX} and the following lemma:

\begin{lemma}\label{-LF_1} The chain $\psi_{k}$ satisfies all conditions
of theorem \ref{chainConvThCommonX} with any finite measure $\mu$
equivalent to the Lebesgue measure on $L$ and $f(\psi)=H(\psi)$.
\end{lemma}

Now we will prove this lemma.

\subsection{Proof of weak Feller property}

Let $O$ be any fixed open subset of $L$. For any $\psi\in L$ denote
$\mathbf{1}_{\psi}(\tau,y),\ \tau\in{\bf R_{+}},\ y\in\mathbb{{\bf R}}^{l}$
the indicator function on ${\bf R_{+}}\times\mathbb{{\bf R}}^{l}$,
that is $\mathbf{1}_{\psi}(\tau,y)=1$ if $J(\tau,y;\psi)\in O$,
and zero otherwise. Then we have 
\[
P(\psi,O)=\int_{\mathbb{{\bf R}}_{+}\times\mathbb{{\bf R}}^{l}}\mathbf{1}_{\psi}(s,y)\rho_{\tau}(s)\rho_{\xi}(y)\ ds\ dy.
\]
Let $\psi_{n}\rightarrow\psi,\psi_{n}\in L$ as $n\rightarrow\infty$.
Fix $s\geqslant0$ and $y\in\mathbb{{\bf R}}^{l}$ and consider two
cases:

1. $J(s,y;\psi)\in O$, then, as $O$ is open and $J(s,y;\psi)$ is
continuous in $\psi$, starting from some $n$ the inclusion $J(s,y;\psi_{n})\in O$
holds. That is why 
\[
\lim_{n\rightarrow\infty}\mathbf{1}_{\psi_{n}}(s,y)=\mathbf{1}_{\psi}(s,y)=1.
\]

2. $J(s,y;\psi)\notin O$. Then 
\[
\liminf_{n}\mathbf{1}_{\psi_{n}}(s,y)\geqslant\mathbf{1}_{\psi}(s,y)=0.
\]
Thus for any $s\geqslant0$ and $y\in\mathbb{{\bf R}}^{l}$ 
\[
\liminf_{n}\mathbf{1}_{\psi_{n}}(s,y)\geqslant\mathbf{1}_{\psi}(s,y).
\]
Then by Fatou's lemma 
\[
\liminf_{n}P(\psi_{n},O)\geqslant P(\psi,O).
\]
Thus the chain satisfies the weak Feller property.

\subsection{Proof of the drift condition}

The drift condition for the chain $\psi_{k}$ of Theorem \ref{chainConvThCommonX}
immediately follows from the following lemma (stronger than Lemma
\ref{-LF_1}).

\begin{lemma} \label{LF_2} There exists compact set $K'\subset L$,
such that for all $\psi\notin K'$ the following inequality holds:
\[
E\{H(\psi_{1})|\psi_{0}=\psi\}-H(\psi)<-rH(\psi),
\]
for some positive constant $r>0$, not depending on $\psi$.
\end{lemma}

\begin{proof} 
It will be convenient to use the following notation: 
\begin{align*}
\psi&= (q,p)^{T},\quad q=(q_{1},\ldots,q_{N})^{T},\quad p=(p_{1},\ldots,p_{N})^{T},\\
\psi^{0}(t)&=  e^{tA}\psi=(q^{0}(t),p^{0}(t))^{T},\\
 q^{0}(t)&=(q_{1}^{0}(t),\ldots,q_{N}^{0}(t))^{T}, \quad p^0=(p_{1}^{0}(t),\ldots, p_{N}^{0}(t))^{T},\ t\geqslant 0.
\end{align*}
Then the point $\psi_{1}$ can be written as 
\[
\psi_{1}=J(\tau,\xi;\psi),
\]
where $\xi=\xi_{1},\ \tau=\tau_{1}$ are random variables defined
above in the definition of $\psi(t)$. For the energy we have: 
\[
H(\psi_{1})=H(\psi^{0}(\tau))+\frac{1}{2}\left(|J(\xi;p_{1}^{0}(\tau))|^{2}-|p_{1}^{0}(\tau)|^{2}\right)=
\]
\[
=H(\psi)+\frac{1}{2}\left(|J(\xi;p_{1}^{0}(\tau))|^{2}-|p_{1}^{0}(\tau)|^{2}\right).
\]
And for the mean energy: 
\[
E\{H(\psi_{1})|\psi_{0}=\psi\}=H(\psi)+\frac{1}{2}(E\{|J(\xi;p_{1}^{0}(\tau))|^{2}\}-E\{|p_{1}^{0}(\tau)|^{2}\}.
\]
By Condition 2 (item 3) on $J$ we have the inequality: 
\[
E\{|J(\xi;p_{1}^{0}(\tau))|^{2}\}\leqslant M+\alpha E\{|p_{1}^{0}(\tau)|^{2}\},
\]
where 
\[
M=\sup_{p\in K}E\{|J(\xi;p)|^{2}\}
\]
and the compact set $K$ was defined in item 3 of Condition 2 on $J$.

Again 
\[
E\{H(\psi_{1})|\psi_{0}=\psi\}\leqslant H(\psi)+\frac{M}{2}-c\Delta(\psi),\quad c=\frac{1-\alpha}{2}>0,
\]
where we used the notation 
\[
\Delta(\psi)=E\{|p_{1}^{0}(\tau)|^{2}\}.
\]
From the definition of $\Delta$ we conclude that it is a non negative
definite quadratic form of $\psi$. On the other hand as the completeness
condition holds and due to the Proposition \ref{prop61} (see below) we have $\Delta(\psi)=0$
iff $\psi=0$. So form $\Delta(\psi)$ is positive definite. As all
positive definite forms define equivalent norms, there exists constant
$r>0$ such that for all $\psi\in L$ the following inequality holds:
\[
\Delta(\psi)>rH(\psi).
\]
It is clear that for some compact subset $K'\subset L$ and for all
$\psi\notin K'$ we have 
\[
M<rcH(\psi).
\]
That is why, for all $\psi\notin K'$ : 
\[
E\{H(\psi_{1})|\psi_{0}=\psi\}-H(\psi)\leqslant-\frac{rc}{2}H(\psi).
\]
Thus the lemma is proved.
\end{proof}

\subsection{Proof of the strong irreducibility property}

For $\psi\in L$, any $m=1,2,\ldots$ and arbitrary $t_{1},\ldots,t_{m}\geqslant0,\ u_{1},\ldots,u_{m}\in\mathcal{O}_{\xi},\ \psi\in L$
define the map $J_{m}$: 
\[
J_{0}=J_{0}(\psi)=\psi,\quad J_{1}(t_{1},u_{1};\psi)=J(t_{1},u_{1};\psi),
\]
\[
\quad J_{2}(t_{1},u_{1},t_{2},u_{2};\psi)=J(t_{2},u_{2};J_{1}(t_{1},u_{1};\psi)),
\]
\[
J_{m}(t_{1},u_{1},\ldots,t_{m},u_{m};\psi)=J(t_{m},u_{m};J_{m-1}(t_{1},u_{1},\ldots,t_{m-1},u_{m-1};\psi)).
\]
Also for any point $\psi\in L$ and any integer $m\geqslant1$ define
the subset: 
\[
\mathcal{J}_{m}(\psi)=\{J_{m}(t_{1},u_{1},\ldots,t_{m},u_{m};\psi):\ t_{1},\ldots,t_{m}\geqslant0,\ u_{1},\ldots,u_{m}\in\mathcal{O}_{\xi}\}\subset L.
\]

Now we prove the following ``Strong controllability Theorem''.

\begin{theorem} \label{overallL} There exists $m\geqslant1$ such
that for any $\psi\in L$ 
\[
\mathcal{J}_{m}(\psi)=L.
\]
\end{theorem}

\begin{proof} Due to condition 2 on the transformation $J$, for all $\psi\in L$
and all $t\geqslant0$ there exists $u=u(t,\psi)\in\mathcal{O}_{\xi}$
such that 
\[
J(t,u;\psi)=J(u;e^{tA}\psi)=Ie^{tA}\psi,
\]
where the linear transformation $I$ is a velocity flip of the first
coordinate of particle 1, i.e.\ if $\psi=(q,p)\in L,q=(q_{1},\ldots,q_{N})^{T},\ p=(p_{1},\ldots,p_{N})^{T}$,
then 
\[
I\psi=\psi'=(q',p'),\ q'=q,\ p'_{1,1}=-p_{1,1},
\]
\[
\ p'_{1,k}=p_{1,k},\ k=2,\ldots,d,\ p'_{j}=p_{j},\ j=2,\ldots,N.
\]
Thus, due to Theorem 1 of \cite{LM_6} there exists $m$ such that
for all $\psi$ 
\[
\mathcal{M}_{h}\subset\mathcal{J}_{m}(\psi),
\]
where $h=H(\psi)$ and 
\[
\mathcal{M}_{h}=\{\psi\in L:H(\psi)=h\}
\]
is the constant energy manifold.

Consider now any point $\psi'\in L,\ h'=H(\psi')$. Let us denote
$g=(0,p)\in L,\ p=(p_{1},\ldots,p_{N})^{T}$, where 
\[
p_{1,1}=1,\ p_{1,k}=0,\ k=2,\ldots,d,\ p_{j}=0,\ j=2,\ldots,N.
\]
It is obvious that $\sqrt{2h}g\in\mathcal{M}_{h}$, and therefore
there exist $t_{1},u_{1},\ldots,t_{m},u_{m}$ such that 
\[
J_{m}(t_{1},u_{1},\ldots,t_{m},u_{m};\psi)=\sqrt{2h}g\in\mathcal{M}_{h}.
\]
Due to the second part of condition 2 on the transformation $J$ we
can find $u_{m+1},t_{m+1}=0$ such that 
\[
J_{m+1}(t_{1},u_{1},\ldots,t_{m},u_{m},t_{m+1},u_{m+1};\psi)=J(0,u_{m+1};\sqrt{2h}g)=\tilde{\psi},
\]
where $\tilde{\psi}=(0,\tilde{p}),\ \tilde{p}=(\tilde{p}_{1},\ldots,\tilde{p}_{N})^{T}$
and 
\[
\tilde{p}_{k}=0,\ k=2,\ldots,N,\ |\tilde{p}_{1}|=\sqrt{2h'},
\]
and so $H(\tilde{\psi})=h'$. Similar arguments show the existence
of $t'_{1},u'_{1},\ldots,t'_{m},u'_{m}$ such that 
\[
J_{m}(t'_{1},u'_{1},\ldots,t'_{m},u'_{m};\tilde{\psi})=\psi'.
\]
Finally we get 
\[
J_{2m+1}(t_{1},u_{1},\ldots,t_{m},u_{m},t_{m+1},u_{m+1},t'_{1},u'_{1},\ldots,t'_{m},u'_{m};\psi)=\psi'.
\]
So, the assertion is proved.
\end{proof}

Let us come back to the proof of the strong irreducibility property.
Fix the number $m$ from Theorem \ref{overallL} and consider the
following set: 
\[
\mathcal{U}=\{(u_{1},t_{1},\ldots,u_{m},t_{m}):\ t_{i}\geqslant0,\ u_{i}\in\mathcal{O}_{\xi}\ \mbox{for all}\ i=1,\ldots,m\}=(\mathcal{O}_{\xi}\times\mathbb{R}_{\geqslant0})^{m}.
\]
When $\psi\in L$ is fixed, we can consider the map $J_{m}(u_{1},t_{1},\ldots,u_{m},t_{m};\psi)$
as a map from $\mathcal{U}$ to $L$. We will denote this map 
\[
G^{\psi}:\mathcal{U}\rightarrow L,\quad\mathbf{u}=(u_{1},t_{1},\ldots,u_{m},t_{m})\mapsto J_{m}(u_{1},t_{1},\ldots,u_{m},t_{m};\psi).
\]
Theorem \ref{overallL} states that the image of the map $G^{\psi}$
coincides with $L$ for all $\psi\in L$. Denote $\lambda$ and $\mu$
Lebesgue measures on $\mathcal{U}$ and $L$ accordingly.

\begin{lemma} \label{zeroSet-1} For any measurable $B\subset L$
its Lebesgue measure $\mu(B)=0$ iff the Lebesgue measure $\lambda$
of the set $(G^{\psi})^{-1}(B)$ in $\mathcal{U}$ is zero.
\end{lemma}

The proof of this lemma is the same as the proof of lemma 7 in
\cite{LM_6}. But for reader's convenience we will give the proof
here in our notation.

\begin{proof} 1) Assume that for some $B\subset L$ we have $\mu(B)=0$.
Let us show that $\lambda((G^{\psi})^{-1}(B))=0$. Let $A_{cr}$ be
the set of critical points of the map $G^{\psi}$ (that is points
$\mathbf{u}=(u_{1},t_{1},\ldots,u_{m},t_{m})\in\mathcal{U}$ where the
rank of the Jacobin is not maximal) and let $E=G^{\psi}(A_{cr})\subset L$
be the set of critical values of $G^{\psi}$. By Sard's theorem $\mu(E)=0$.
But as $G^{\psi}(\mathcal{U})=L$, then there exists a non-critical
point $\mathbf{u}\in\mathcal{U}$, that is such that the rank of $dG^{\psi}$
at this point equals $2N$. As the map $G^{\psi}$ is analytic in
the variables $u_{1},t_{1},\ldots,u_{m},t_{m}$, the set of points
$A_{cr}$, where the rank is less than $2N$, has Lebesgue measure
zero. Then the equality $\lambda((G^{\psi})^{-1}(B))0$ follows from
the Ponomarev's theorem (theorem 1 of \cite{Ponomarev}).

2) Assume that for some $B\subset L$ we have $\mu(B)>0$, and let
us show that $\lambda((G^{\psi})^{-1}(B))>0$. By Lebesgue differentiation
theorem there exists point $\psi'\in L\setminus E$ and its neighborhood
$O(\psi')$ such that $\mu(O(\psi')\cap B)>0$. Then there is point
$\mathbf{u}=\mathbf{u}(\psi')\in(G^{\psi})^{-1}(\psi')$ and some its
neighborhood $O(\mathbf{u})\subset\mathcal{U}$ such that the restriction
of $G^{\psi}$ on $O(\mathbf{u})$ is a submersion. Then $\mu(O(\psi')\cap B)>0$
implies $\lambda((G^{\psi})^{-1}(B)\cap O(\mathbf{u}))>0$. So the lemma
is proved.
\end{proof}

Consider the probability density function of $\xi_{1},\tau_{1},\ldots,\xi_{m},\tau_{m}$
w.r.t.\ Lebesgue measure on $\mathcal{U}$: 
\[
p(\mathbf{u})=\rho_{\xi}(u_{1})\rho_{\tau}(t_{1})\ldots\rho_{\xi}(u_{m})\rho_{\tau}(t_{m}),\quad\mathbf{u}=(u_{1},t_{1},\ldots,u_{m},t_{m})\in\mathcal{U}.
\]
We have an obvious equality: 
\[
P^{m}(\psi,B)=\int_{(G^{\psi})^{-1}(B)}p(\mathbf{u})d\mathbf{u}.
\]
Due to Lemma \ref{zeroSet-1} and this equality we conclude that
the strong irreducibility property holds for the chain $\psi_{k}$.

\section{Proof of Theorem \ref{thConvProcess}}

Theorem \ref{thConvProcess} follows from ergodicity of the embedded
chain $\psi_{k}$ (Theorem \ref{thConvChain}). Proof of this statement
 coincides almost verbatim with the proof of similar statement in
\cite{LM_6} (theorem 2 from the section ``Proof of theorem 2'').

\section{Proof of Theorem \ref{limitMeasurePropertyTh} }

In this case the process $\psi(t)$ is a time homogeneous Markov process.
Denote by $\mathcal{A}$ the infinitesimal generator of $\psi(t)$, i.e.
\[
(\mathcal{A}f)(\psi)=\lim_{t\rightarrow0}\frac{E_{\psi}f(\psi(t))-f(\psi)}{t}.
\]
The well known formula is 
\begin{equation}
(\mathcal{A}f)(\psi)=\{H,f\}+\lambda I(f),\label{infGenerator}
\end{equation}
where 
\[
I(f)=E_{\xi}f(J(\xi;\psi))-f(\psi)=E_{v}f(q;\alpha p_{1}+(1-\alpha)Mv,p_{2},\ldots,p_{N})-f(q;p)
\]
and 
\[
\{H,f\}=\sum_{k=1}^{N}\frac{\partial H}{\partial{p_{k}}}\frac{\partial f}{\partial{q_{k}}}-\frac{\partial H}{\partial{q_{k}}}\frac{\partial f}{\partial{p_{k}}}
\]
is a Poisson bracket.

Let us prove the first part of Theorem \ref{limitMeasurePropertyTh}.
Assume that $p_{\beta}(\psi)$ is the Gibbs probability density of
the invariant measure for the process $\psi(t)$ for some $\beta>0$.
Thus $\mathcal{A}^{*}p_{\beta}(\psi)=0$ for all $\psi\in L$ where
$\mathcal{A}^{*}$ is a formal adjoint operator to $\mathcal{A}$
w.r.t.\ the standard inner product on $L_{2}(d\psi)$. For the bracket
we have 
\[
\{H,\cdot\}^{*}=-\{H,\cdot\}.
\]
Consequently, for all $\beta>0$ 
\[
\{H,p_{\beta}\}^{*}=0.
\]
Let us find $I^{*}$. For any functions $f,g\in L_{2}$ we have: 
\begin{align*}
 & \int E_{v}\{f(q;\alpha p_{1}+(1-\alpha)Mv,p_{1},\ldots,p_{N})\}g(q,p)\ dqdp=\\[5pt]
 & =\frac{1}{\alpha}\int f(q,p)E_{v}\Bigl\{g(q;\frac{p_{1}-(1-\alpha)Mv}{\alpha},p_{2},\ldots,p_{N})\Bigr\}\ dqdp.
\end{align*}
That is why 
\[
I^{*}(f)=\gamma E_{v}\{f(q;\gamma p_{1}-(1-\gamma)Mv,p_{2},\ldots,p_{N})\}-f(q,p),\quad \gamma=\frac{1}{\alpha} .
\]
In case $\alpha=0$ one should take the formal limit in this formula
for $I^{*}(f)$, then we will have $I^{*}(f)=1$.

If $f=p_{\beta}$, then 
\[
E_{v}\{p_{\beta}(q;\gamma p_{1}-(1-\gamma)Mv,p_{2},\ldots,p_{N})\}=
\]
\begin{equation}
  =\frac{1}{Z}e^{-\beta H(q,p)}\exp\Bigl\{\frac{\beta}{2M}p_{1}^{2}\Bigr\} E_{v}
  \Bigl\{\exp\Bigl(-\frac{\beta}{2M}(\gamma p_{1}-(1-\gamma)Mv)^{2}\Bigr)\Bigr\} . \label{IAdjPb}
\end{equation}

\begin{lemma}\label{-L_gaussian} If $\eta$ is a gaussian random
variable with mean $a$ and variance $\sigma^{2}$ then 
\[
E\exp(-\eta^{2})=\sqrt{\frac{1}{1+2\sigma^{2}}}\exp\left(-\frac{a^{2}}{1+2\sigma^{2}}\right).
\]
\end{lemma}
\begin{proof}
  By definition we have 
\[
E\exp(-\eta^{2})=\frac{1}{\sqrt{2\pi\sigma^{2}}}\int_{\mathbb{R}}e^{-x^{2}}e^{-(x-a)^{2}/(2\sigma^{2})}dx.
\]
On the other hand 
\[
x^{2}+\frac{(x-a)^{2}}{2\sigma^{2}}=x^{2}\Bigl(1+\frac{1}{2\sigma^{2}}\Bigr)-x\frac{a}{\sigma^{2}}+\frac{a^{2}}{2\sigma^{2}}=\vphantom{\int}
\]
\[
  =\Bigl(1+\frac{1}{2\sigma^{2}}\Bigr)\Bigl(x-\frac{a}{2\sigma^{2}}\frac{1}{1+\frac{1}{2\sigma^{2}}}\Bigr)^{2}+
  \frac{a^{2}}{2\sigma^{2}}-\frac{a^{2}}{4\sigma^{4}}\frac{1}{1+\frac{1}{2\sigma^{2}}}=\vphantom{\int}
\]
\[
  =\Bigl(1+\frac{1}{2\sigma^{2}}\Bigr)\Bigl(x-\frac{a}{2\sigma^{2}}\frac{1}{1+\frac{1}{2\sigma^{2}}}\Bigr)^{2}+
  \frac{a^{2}}{1+2\sigma^{2}}=\frac{(x-\frac{a}{2\sigma^{2}+1})^{2}}{2\frac{\sigma^{2}}{1+2\sigma^{2}}}+\frac{a^{2}}{1+2\sigma^{2}}.\vphantom{\int^A}
\]
Thus, 
\begin{align*}
  E\exp(-\eta^{2})&=\frac{1}{\sqrt{2\pi\sigma^{2}}}\sqrt{2\pi\frac{\sigma^{2}}{1+2\sigma^{2}}}
                    \exp\Bigl(-\frac{a^{2}}{1+2\sigma^{2}}\Bigr)\\
  &=\sqrt{\frac{1}{1+2\sigma^{2}}}\exp\Bigl(-\frac{a^{2}}{1+2\sigma^{2}}\Bigr) .
\end{align*}
So the lemma is proved.
\end{proof}

Now assume that the velocity of the external particle $v$ is a zero
mean gaussian random variable with the variance $\sigma_{v}^{2}$.
From the equation (\ref{IAdjPb}) and lemma \ref{-L_gaussian} we
have 
\begin{align*}
  &E_{v}\{p_{\beta}(q;\gamma p_{1}-(1-\gamma)Mv,p_{2},\ldots,p_{N})\}\\
  &\quad =\frac{1}{Z}e^{-\beta H(q,p)}\exp\Bigl\{\frac{\beta}{2M}p_{1}^{2}\Bigr\}\sqrt{\frac{1}{1+2\sigma^{2}}}\exp\Bigl(-\frac{a^{2}}{1+2\sigma^{2}}\Bigr),
\end{align*}
where 
\[
a=E_{v}\sqrt{\frac{\beta}{2M}}(\gamma p_{1}-(1-\gamma)Mv)=\sqrt{\frac{\beta}{2M}}\gamma p_{1},
\]
\[
\sigma^{2}=D\sqrt{\frac{\beta}{2M}}(\gamma p_{1}-(1-\gamma)Mv)=\frac{\beta M}{2}(1-\gamma)^{2}\sigma_{v}^{2}.
\]
So, 
\[
E_{v}\{p_{\beta}(q;\gamma p_{1}-(1-\gamma)v,p_{2},\ldots,p_{N})\}=
\]
\[
    =\frac{1}{Z}e^{-\beta H(q,p)}\sqrt{\frac{1}{1+2\sigma^{2}}}\exp\Bigl(\frac{\beta}{2M}p_{1}^{2}
    \Bigl(1-\frac{\gamma^{2}}{1+\beta M(1-\gamma)^{2}\sigma_{v}^{2}}\Bigr)\Bigr).
\]
Choose $\beta$ such that 
\[
1-\frac{\gamma^{2}}{1+\beta M(1-\gamma)^{2}\sigma_{v}^{2}}=0,
\]
i.e. 
\[
\beta=\frac{\gamma^{2}-1}{M(1-\gamma)^{2}\sigma_{v}^{2}}=\frac{\gamma+1}{M(\gamma-1)\sigma_{v}^{2}}=\frac{1+\alpha}{M(1-\alpha)\sigma_{v}^{2}}.
\]

For this $\beta$ we have $\sigma^{2}=(\gamma^{2}-1)/2$ and 
\[
E_{v}\{p_{\beta}(q;\gamma p_{1}-(1-\gamma)Mv,p_{2},\ldots,p_{N})\}=\frac{1}{Z}e^{-\beta H(q,p)}\frac{1}{\gamma}.
\]
That is why 
\[
\mathcal{A}^{*}p_{\beta}=I^{*}p_{\beta}=0
\]
and so the Gibbs measure is an invariant for the process $\psi(t)$.

Conversely, let $p_{\beta}$ satisfy the equation $\mathcal{A}^{*}p_{\beta}=0$
for some $\beta>0$. Now we prove that in this case the distribution
of $v$ is gaussian with zero mean. From the (\ref{IAdjPb}) we have
that for all $p_{1}\in\mathbb{{\bf R}}$: 
\begin{equation}
\gamma E_{v}\Bigl\{\exp\Bigl(-\frac{\beta}{2M}(\gamma p_{1}-(1-\gamma)Mv)^{2}\Bigr)\Bigr\}=\exp\Bigl\{-\frac{\beta}{2M}p_{1}^{2} \Bigr\}. \label{rhoveq}
\end{equation}
Multiplying the both side of the last equality on $e^{i\lambda p_{1}}$
and integrating by $p_{1}$ over the $\mathbb{R}^{1}$ we obtain 
\[
    \int_{\mathbb{{\bf R}}^{1}}e^{i\lambda p_{1}}E_{v}\Bigl\{\exp\Bigl(-\frac{\beta}{2M}(\gamma p_{1}-(1-\gamma)Mv)^{2}\Bigr)\Bigr\}dp_{1}
    =\int_{\mathbb{{\bf R}}^{1}}e^{i\lambda p_{1}}e^{-\frac{\beta}{2M}p_{1}^{2}}dp_{1}.
\]
The right hand side is the characteristic function of the zero mean
gaussian random variable with the variance $\sqrt{M / \beta}$
and so 
\begin{equation}
\int_{\mathbb{{\bf R}}^{1}}e^{i\lambda p_{1}}e^{-\frac{\beta}{2M}p_{1}^{2}}dp_{1}=\sqrt{\frac{2\pi M}{\beta}}e^{-\frac{M\lambda^{2}}{2\beta}}. \label{rightp1expi}
\end{equation}
If we denote by $\rho$ the density of $v$ then for the left hand side
we have 
\[
l(\lambda)=\int_{\mathbb{{\bf R}}^{1}}e^{i\lambda p_{1}}E_{v}\Bigl\{\exp\Bigl(-\frac{\beta}{2M}(\gamma p_{1}-(1-\gamma)Mv)^{2}\Bigr)\Bigr\}dp_{1}=
\]
\[
=\int_{\mathbb{{\bf R}}^{1}}dv\rho(v)\int_{\mathbb{{\bf R}}^{1}}dp_{1}e^{i\lambda p_{1}}\exp\Bigl(-\frac{\beta}{2M}(\gamma p_{1}-(1-\gamma)Mv)^{2}\Bigr).
\]
The integral 
\[
\int_{\mathbb{{\bf R}}^{1}}e^{i\lambda p_{1}}\exp\Bigl(-\frac{\beta}{2M}(\gamma p_{1}-(1-\gamma)Mv)^{2}\Bigr)dp_{1}
\]
is the characteristic function of the gaussian random variable with
mean $a=(1-\gamma)Mv/\gamma$ and variance $\sigma^{2}=\sqrt{M / (\beta\gamma^{2})}$.
That is why 
\[
    \int_{\mathbb{{\bf R}}^{1}}e^{i\lambda p_{1}}\exp\Bigl(-\frac{\beta}{2M}(\gamma p_{1}-(1-\gamma)Mv)^{2}\Bigr) dp_{1}=\sqrt{\frac{2\pi M}{\beta\gamma^{2}}}
    \exp\Bigl(i\lambda a-\frac{\sigma^{2}\lambda^{2}}{2}\Bigr)=
\]
\[
=\sqrt{\frac{2\pi M}{\beta\gamma^{2}}}\exp\Bigl(i\lambda\frac{(1-\gamma)Mv}{\gamma}-\frac{M\lambda^{2}}{2\beta\gamma^{2}}\Bigr).
\]
Then 
\[
    l(\lambda)=\sqrt{\frac{2\pi M}{\beta\gamma^{2}}}\exp\Bigl(-\frac{M\lambda^{2}}{2\beta\gamma^{2}}\Bigr)
    \int_{\mathbb{R}^{1}}\exp\Bigl(i\lambda\frac{(1-\gamma)Mv}{\gamma}\Bigr)\rho(v)dv=
\]
\[
=\sqrt{\frac{2\pi M}{\beta\gamma^{2}}}\exp\Bigl(-\frac{M\lambda^{2}}{2\beta\gamma^{2}}\Bigr)\phi\Bigl(\frac{(1-\gamma)M\lambda}{\gamma}\Bigr),
\]
where $\phi(\lambda)$ is a characteristic function of $v$: 
\[
\phi(\lambda)=\int_{\mathbb{{\bf R}}}e^{i\lambda v}\rho(v)dv.
\]
Consequently from the (\ref{rhoveq}) and (\ref{rightp1expi}) for
all $\lambda\in\mathbb{{\bf R}}$ we have 
\[
\phi\Bigl(\frac{(1-\gamma)M\lambda}{\gamma}\Bigr)=\exp\Bigl(-\frac{M\lambda^{2}}{2\beta}+\frac{M\lambda^{2}}{2\beta\gamma^{2}}\Bigr).
\]
Immediately we conclude that the distribution of $v$ is zero mean
gaussian. So the first part of Theorem \ref{thConvProcess} is
    proved.

    Now we prove the second part.
Let us denote 
\[
Q_{k}(t)=E_{\psi}q_{k}(t)=E\{q_{k}(t)|\psi(0)=\psi\},\;\, P_{k}(t)=E_{\psi}p_{k}(t)=E\{p_{k}(t)|\psi(0)=\psi\}.
\]
Then using formula (\ref{infGenerator}) for the infinitesimal generator
we get the following ordinary differential equations 
\begin{align*}
\dot{Q}_{k}= & \ \frac{1}{M}P_{k},\\
\dot{P}_{k}= & \ -\sum_{j=1}^{N}V_{k,j}Q_{j}-\lambda(1-\alpha)P_{1}(t)\delta_{k,1},
\end{align*}
for all $k=1,\ldots,N$, where $\delta_{k,1}$ is a Kronecker symbol
and the initial conditions are given by $(Q(0),P(0))=\psi$. If we
return to the velocity variables $V_{k}=P_{k}M$, then we get 
\begin{align*}
\dot{Q}_{k}= & \ V_{k},\\
\dot{V}_{k}= & \ -\frac{1}{M}\sum_{j=1}^{N}V_{k,j}Q_{j}-\lambda(1-\alpha)V_{1}(t)\delta_{k,1},
\end{align*}
for all $k=1,\ldots,N$.  Due to theorem 2.1 from \cite{LM_1}
we obtain 
\[
\lim_{t\rightarrow\infty}Q_{k}(t)=\lim_{t\rightarrow\infty}V_{k}(t)=0,
\]
for all $k=1,\ldots,N$ and any initial condition $\psi\in L$, i.e.
\begin{equation}
\lim_{t\rightarrow\infty}E_{\psi}\psi(t)=0,\label{meanPsiLim}
\end{equation}
for any initial $\psi(0)\in L$. Now let us consider the matrix 
\[
C(t)=E_{\psi}\psi(t)\psi^{T}(t).
\]

Denote by $\Gamma$ the $2N\times2N$ matrix all elements of which are
equal to $0$ except $(2N+1,2N+1)$-element which is equal to $1$.

\begin{lemma} The matrix $C(t)$ satisfies the following differential
equation 
\begin{equation}
\dot{C}(t)=AC+CA^{T}-\lambda(1-\alpha)(\Gamma C+C\Gamma-(1-\alpha)\Gamma C\Gamma)+\lambda(1-\alpha)^{2}M^{2}\sigma^{2}\Gamma \label{dotCeq}
\end{equation}
where matrix $A$ was defined in (\ref{LHamVec}).
\end{lemma}

\begin{proof}
  We will use formula (\ref{infGenerator}) for the infinitesimal
generator. To get the $\{H,C\}$-term note that it can be received
as a derivative of $\psi(t)\psi^{T}(t)$ by equation (\ref{LHamVec}).
In other words, if $\psi_{0}(t)$ is a solution of (\ref{LHamVec})
then 
\[
\frac{d}{dt}\psi_{0}(t)\psi_{0}^{T}(t)=A\psi_{0}(t)\psi_{0}^{T}(t)+\psi_{0}(t)\psi_{0}^{T}(t)A^{T}=\{H,\psi_{0}(t)\psi_{0}^{T}(t)\}.
\]
Now we want to find the term $I(f)$ in (\ref{infGenerator}).

Denote by $B$ the $(2N\times2N)$-diagonal matrix which has all diagonal
elements equal to $1$, except $(2N+1,2N+1)$-element which is equal
to $\alpha$ and $g$ is $2N$-vector with all entries equal to zero
except $(2N+1$)-element equal to $1$. Then $\Gamma=gg^{T}$ and
the transformation $p_{1}\to\alpha p_{1}+(1-\alpha)Mv$ can be written
as 
\[
\psi\to B\psi+(1-\alpha)Mvg.
\]
So we get 
\[
I(\psi\psi^{T})=E_{v}(B\psi+(1-\alpha)Mvg)(B\psi+(1-\alpha)Mvg)^{T}-\psi\psi^{T}=
\]
\[
=B\psi\psi^{T}B+(1-\alpha)^{2}M^{2}\sigma^{2}gg^{T}-\psi\psi^{T}.
\]
Consequently, we have the equation: 
\[
\dot{C}=AC+CA^{T}+\lambda(BCB-C)+\lambda(1-\alpha)^{2}M^{2}\sigma^{2}\Gamma.
\]
The matrix $B$ can be written as $B=E-(1-\alpha)gg^{T}$. So 
\[
BCB-C=-(1-\alpha)(\Gamma C+C\Gamma)+(1-\alpha)^{2}\Gamma C\Gamma .
\]
The lemma is proved.
\end{proof}

Let us introduce notation for the linear part of the right hand side
of equation (\ref{dotCeq}): 
\[
L(C)=AC+CA^{T}-\lambda(1-\alpha)(\Gamma C+C\Gamma-(1-\alpha)\Gamma C\Gamma).
\]
Then the right hand side of the equation (\ref{dotCeq}) can be rewritten
as 
\[
\dot{C}=L(C)+\lambda(1-\alpha)^{2}M^{2}\sigma^{2}\Gamma.
\]
Denote 
\[
C_{G}=\left(\begin{array}{cc}
V^{-1} & 0\\
0 & ME
\end{array}\right)
\]
the Gibbs matrix. Then 
\[
A=\left(\begin{array}{cc}
0 & \tfrac{1}{M}E\\[1pt]
-V & 0
\end{array}\right)
\]
and we have 
\[
AC_{G}+C_{G}A^{T}=AC_{G}+(AC_{G})^{T}=0,
\]
\[
\Gamma C_{G}+C_{G}\Gamma=2M\Gamma,\quad\Gamma C_{G}\Gamma=M\Gamma.
\]
Thus, for $\beta>0$ 
\[
L(\beta^{-1}C_{G})+\lambda(1-\alpha)^{2}M^{2}\sigma^{2}\Gamma\! =\! -\lambda(1-\alpha)M(2\beta^{-1}-(1-\alpha)\beta^{-1}\! -\! (1-\alpha)M\sigma^{2})\Gamma.
\]
If we choose $\beta>0$ so that 
\[
2\beta^{-1}-(1-\alpha)\beta^{-1}-(1-\alpha)M\sigma^{2}=0,
\]
i.e. 
\[
\beta^{-1}=\frac{1-\alpha}{1+\alpha}M\sigma^{2},
\]
then 
\[
L(\beta^{-1}C_{G})+\lambda(1-\alpha)^{2}M^{2}\sigma^{2}\Gamma=0
\]
and we conclude that 
\[
C_{G,\beta}=\beta^{-1}C_{G}
\]
is a fixed point of the equation (\ref{dotCeq}).

\begin{lemma} For any initial non-negative definite $(2N\times2N)$-matrix
$C(0)$ the solution of the homogeneous equation 
\begin{equation}
\dot{C}(t)=L(C(t)) \label{homogCeq}
\end{equation}
converges to zero as $t\rightarrow\infty$.
\end{lemma}

\begin{proof}
  For the proof we want to find a corresponding Lyapunov function.
Consider the function 
\[
F(C)=\mathrm{Tr}(C_{G}^{-1}C),
\]
where $\mathrm{Tr}$ denotes the trace. Note that if $C$ is non-negative
definite then $F(C)\geqslant0$ and $F(C)=0$ iff $C=0$. We will
prove that $F(C)$ is non-increasing along the trajectories of the
equation (\ref{homogCeq}). If $C(t)$ a solution of the equation
(\ref{homogCeq}) we have 
\[
\frac{d}{dt}F(C(t))=\mathrm{Tr}(C_{G}^{-1}L(C))=
\]
\[
  =\mathrm{Tr}(C_{G}^{-1}AC+C_{G}^{-1}CA^{T}) - 
\]
\[
  -\lambda(1-\alpha)\mathrm{Tr}(C_{G}^{-1}\Gamma C+C_{G}^{-1}C\Gamma-(1-\alpha)C_{G}^{-1}\Gamma C\Gamma)
  \mathrm{Tr}(C_{G}^{-1}AC+C_{G}^{-1}CA^{T})=
\]
\[
  = \mathrm{Tr}((C_{G}^{-1}A+A^{T}C_{G}^{-1})C)=0.
\]
\[
\mathrm{Tr}(C_{G}^{-1}\Gamma C+C_{G}^{-1}C\Gamma-(1-\alpha)C_{G}^{-1}\Gamma C\Gamma)=
\]
\[
=\mathrm{Tr}((C_{G}^{-1}\Gamma+\Gamma C_{G}^{-1}-(1-\alpha)\Gamma C_{G}^{-1}\Gamma)C)=
\]
\[
=\Bigl(\frac{2}{M}-\frac{(1-\alpha)}{M}\Bigr)\mathrm{Tr}(\Gamma C)=\frac{(1+\alpha)}{M}C_{2N+1,2N+1}.
\]
Thus 
\begin{equation}
\frac{d}{dt}F(C(t))=-\lambda\frac{(1-\alpha^{2})}{M}C_{2N+1,2N+1}(t)\leqslant0.\label{LyapFuncC}
\end{equation}
Further, we want to use Barbashin\tire Krasovskij's theorem (\hspace{-0.1pt}\cite{Barbashin},
p.\thinspace 19, Th.\thinspace 3.2). For this we need to check that the set of the non-negative
definite matrices $C$, whose $(2N+1,2N+1)$-element is zero, does
not contain the solution of (\ref{homogCeq}) except for the zero
solution. Assume the contrary, i.e.\ that there exists solution $C(t)$
of (\ref{homogCeq}) such that $C_{2N+1,2N+1}(t)=0$ for all $t\geqslant0$.
For such solution we have $\Gamma C(t)\Gamma=0$. Thus 
\[
\dot{C}=AC+CA^{T}-\lambda(1-\alpha)(\Gamma C+C\Gamma)=A_{D}C+CA_{D}^{T},
\]
where $A_{D}=A-\lambda(1-\alpha)\Gamma$. It is easy to see that the
solution of the last equation is given by the formula: 
\[
C(t)=\int_{0}^{t}e^{sA_{D}}C(0)e^{sA_{D}^{T}}\ ds+C(0).
\]
Denote by $(,)$ the standard Euclidean inner product on $L$. Then we
have for all $t\geqslant0$ 
\[
0=C_{2N+1,2N+1}(t)=(g,C(t)g)=\int_{0}^{t}(g,e^{sA_{D}}C(0)e^{sA_{D}^{T}}g)\ ds=
\]
\[
=\int_{0}^{t}(e^{sA_{D}^{T}}g,C(0)e^{sA_{D}^{T}}g)\ ds=0.
\]
Since $C(0)$ is non-negative definite, the last equality takes place 
iff $\exp\{tA_{D}^{T}\}g$ $\in\mathrm{Ker}(C(0))$ for all $t\geqslant0$,
where $\mathrm{Ker}C(0)$ is the kernel of $C(0)$. Consequently,
there exists non-zero vector $u\in L$ such that $(e^{tA_{D}^{T}}g,u)=0$
for all $t\geqslant0$. But 
\[
\bigl(e^{tA_{D}^{T}}g,u\bigr)=\sum_{k=0}^{\infty}\frac{t^{k}}{k!}\bigl((A_{D}^{T})^{k}g,u\bigr).
\]
Then $((A_{D}^{T})^{k}g,u)=0$ for all $k=0,1,\ldots$ Due to the
completeness property of $V$ and lemma 3.1 of \cite{LM_1} (see the
proof of the latter lemma) we conclude that $u=0$. This contradicts
our assumptions. So the set of non-negative definite matrices $C$,
whose $(2N+1,2N+1)$-element is zero, does not contain the whole trajectories
of the solution of (\ref{homogCeq}). Then using (\ref{LyapFuncC})
and Barbashin\tire Krasovskij's theorem we get the proof of the lemma.
\end{proof}

Let us come back to the proof of the second part of theorem \ref{limitMeasurePropertyTh}.
Remind what we have already proved. If $\psi(t)$ is our process,
then we have proved that 
\[
\lim_{t\rightarrow\infty}E_{\psi}\psi(t)=0.
\]
Then the matrix $C(t)=E_{\psi}\psi(t)\psi(t)^{T}$ satisfies the equation
(\ref{dotCeq}) which has fixed point $C_{G,\beta}$. Moreover, the
solution of the homogeneous part of ( \ref{dotCeq} ) converges to
zero as $t\rightarrow\infty$ due to the last lemma. All these facts
give us the second proposition of Theorem \ref{limitMeasurePropertyTh}.
So the Theorem is finally proved.

\section{APPENDIX}

\subsection{Collisions}

Here we prove that two simplest collision cases satisfy our general
Condition~$2$.

Moreover, assume that their interaction radius is small, and collision
time is short. That is during this time the potential between our
particle of mass $M$ and other part of the system does not change.

\paragraph{One-dimensional collisions}

Let $d=1$. Assume that a particle of mass $M$ of the $N$-particle
system collides with external particles of mass $m$ and velocity
$v$. Then conservation of energy and momentum give the formula (see
for example in this issue \cite{M_Intro_1}): 
\[
p'=J(v;p)=\alpha p+M(1-\alpha)v,\quad\alpha=\frac{M-m}{M+m},
\]
where $p,p'$ are the momenta of our particle of mass $M$ before
and after collision correspondingly.

\begin{lemma} Assume that the distribution of the velocity $v$ has
density on $R$ and its support coincides with $\mathbb{{\bf R}}$.
Moreover, assume that $a_{2}=Ev^{2}<\infty$. Then Condition 2 
holds. \end{lemma}

\begin{proof} Points 1 and 2 evidently hold. Let us prove point 3:
\[
E_{v}\{J^{2}(v;p)\}=\alpha^{2}p^{2}+2M\alpha(1-\alpha)a_{1}p+M^{2}(1-\alpha)^{2}a_{2},\quad a_{1}=Ev.
\]
But it is clear that for any $\epsilon>0$ there exists $P$ such
that for any $p$ such that $|p|>P$ 
\[
E_{v}\{J^{2}(v;p)\}<(\alpha^{2}+\epsilon)p^{2}.
\]
This gives the proof.
\end{proof}

\paragraph{Two-dimensional non central collision}

Let $d=2$. Consider elastic collision of two ideal two-dimensional
balls with masses $m_{1}$ and $m_{2}$ on ${\bf R}^{2}$. Assume
that at the moment of collision the centers of the balls have coordinates
\[
O_{1}=(x_{1},y_{1}),\quad O_{2}=(x_{2},y_{2}),
\]
and velocities 
\[
\vec{v}_{1}=(v_{1}^{x},v_{1}^{y})^{T},\quad\vec{v}_{2}=(v_{2}^{x},v_{2}^{y})^{T},
\]
(vector here are column vectors). Assume obviously that $O_{1}\ne O_{2}$.
Denote by 
\[
(\vec{v}_{1})'=((v_{1}^{x})',(v_{1}^{y})')^{T},\quad(\vec{v}_{2})'=((v_{2}^{x})',(v_{2}^{y})')^{T}
\]
the velocities of the balls after the collisions.

Define the normalized vector connecting the centers of the balls:
\[
\vec{R}=\frac{1}{|O_{1}O_{2}|}\overrightarrow{O_{1}O_{2}},
\]
where $|O_{1}O_{2}|$ is the distance between points $O_{1}$ and
$O_{2}$.

\begin{lemma} \label{twodimColl} There exists $\phi\in[0,2\pi)$
such that 
\[
\vec{R}=R(\phi)=(\cos\phi,\sin\phi)^{T}.
\]
Then also
\[
(\vec{v}_{1})'=G_{\alpha}(\phi)\vec{v}_{1}+c_{\alpha}(\phi,\vec{v}_{2})R(\phi),
\]
where 
\begin{align*}
G_{\alpha}(\phi)= & \left(\begin{array}{cc}
\alpha\cos^{2}\phi+\sin^{2}\phi & -\frac{1-\alpha}{2}\sin2\phi\\[3pt]
-\frac{1-\alpha}{2}\sin2\phi & \alpha\sin^{2}\phi+\cos^{2}\phi
\end{array}\right),\quad\alpha=\frac{m_{1}-m_{2}}{m_{1}+m_{2}},\\[3pt]
c_{\alpha}(\phi,\vec{v}_{2})= & (1-\alpha)(v_{2}^{x}\cos\phi+v_{2}^{y}\sin\phi)
\end{align*}
\end{lemma}

\begin{proof} Denote by $\vec{R}_{\perp}$ a unit vector (one of the two) orthogonal
to $\vec{R}$. Then we can expand the velocity in terms of tangential
and normal components with respect to the vector $\vec{R}$: 
\begin{align}
 & \vec{v}_{1}=v_{1,n}\vec{R}+v_{1,t}\vec{R}_{\perp},\quad\vec{v}_{2}=v_{2,n}\vec{R}+v_{2,t}\vec{R}_{\perp},\label{beforeCollDec}\\
 & (\vec{v}_{1})'=v'_{1,n}\vec{R}+v'_{1,t}\vec{R}_{\perp},\quad(\vec{v}_{2})'=v'_{2,n}\vec{R}+v'_{2,t}\vec{R}_{\perp}.
\end{align}
Then the conservation laws of energy and momentum are as follows:
\begin{align*}
 & m_{1}v_{1,n}+m_{2}v_{2,n}=m_{1}v'_{1,n}+m_{2}v'_{2,n}\\
 & m_{1}v_{1,t}+m_{2}v_{2,t}=m_{1}v'_{1,t}+m_{2}v'_{2,t}\\
 & m_{1}(v_{1,n}^{2}+v_{1,t}^{2})+m_{2}(v_{2,n}^{2}+v_{2,t}^{2})=m_{1}((v'_{1,n})^{2}\! +\! (v'_{1,t})^{2})\! +\! m_{2}((v'_{2,n})^{2}+(v'_{2,t})^{2}).
\end{align*}
Now we have $3$ equations and four unknowns. But assumptions about
elasticity condition and smoothness of the ball's boundary implies
that the tangential components of the velocities rest unchanged 
\[
v_{1,t}=v'_{1,t},\quad v_{2,t}=v'_{2,t}.
\]
Then, for normal components we get the same equations as in the case
of central elastic collision. Namely 
\[
v'_{1,n}=\alpha v_{1,n}+(1-\alpha)v_{2,n},\quad\alpha=\frac{m_{1}-m_{2}}{m_{1}+m_{2}}.
\]

Now we want to write down velocity transformation after collision
in the initial coordinate system. Note that 
\[
\vec{R}_{\perp}=(-\sin\phi,\cos\phi)^{T}.
\]
which gives: 
\begin{align*}
 & v_{1,n}=v_{1}^{x}\cos\phi+v_{1}^{y}\sin\phi,\quad v_{1,t}=-v_{1}^{x}\sin\phi+v_{1}^{y}\cos\phi,\\
 & v_{2,n}=v_{2}^{x}\cos\phi+v_{2}^{y}\sin\phi,\quad v_{2,t}=-v_{2}^{x}\sin\phi+v_{2}^{y}\cos\phi.
\end{align*}
and 
\begin{align*}
 & (v_{1}^{x})'=v'_{1,n}\cos\phi-v_{1,t}\sin\phi=(\alpha v_{1,n}+(1-\alpha)v_{2,n})\cos\phi-v_{1,t}\sin\phi=\\
 & =v_{1}^{x}(\alpha\cos^{2}\phi+\sin^{2}\phi)-v_{1}^{y}\frac{1-\alpha}{2}\sin2\phi+(1-\alpha)v_{2,n}\cos\phi,\\
 & (v_{1}^{y})'=v'_{1,n}\sin\phi+v_{1,t}\cos\phi=(\alpha v_{1,n}+(1-\alpha)v_{2,n})\sin\phi+v_{1,t}\cos\phi=\\
 & =-v_{1}^{x}\frac{1-\alpha}{2}\sin2\phi+v_{1}^{y}(\alpha\sin^{2}\phi+\cos^{2}\phi)+(1-\alpha)v_{2,n}\sin\phi.
\end{align*}
\end{proof}

For the next statement we assume that the interaction with the external
media is defined by the elastic collision with external particle of
mass $m$ and random velocity $\vec{v}\in\mathbb{{\bf R}}^{2}$. Introduce
the transformation $J$ as follows: 
\[
J(\xi;p)=G_{\alpha}(\phi)p+Mc_{\alpha}(\phi,\vec{v})\vec{R}(\phi),\quad\xi=(\phi,\vec{v}),\quad p\in\mathbb{R}^{2}
\]
where $\alpha=(M-m)/(M+m)$ and the matrix $G_{\alpha}$, vector
$R(\phi)$ and constant $c_{\alpha}(\phi,\vec{v})$ are as in Lemma
\ref{twodimColl}.

\begin{lemma} Assume that the support of the distribution density
of velocity $\vec{v}$ coincides with $\mathbb{{\bf R}}^{2}$, that
the second moment $a_{2}=E|\vec{v}|^{2}$ is finite and the density
support of the angle $\phi$ coincides with $[0,2\pi)$. Then Condition
2 holds for transformation $J$.
\end{lemma}

\begin{proof} The first item evidently holds. For the others we shall
use notation and assertions from the proof of Lemma \ref{twodimColl}.
Consider two arbitrary vectors $\vec{v}_{1},\vec{v}^{*}\in\mathbb{{\bf R}}^{2}$.
To prove point 2 it is sufficient to show that there exists $\phi\in[0,2\pi)$
and $\vec{v}_{2}\in\mathbb{{\bf R}}^{2}$, such that 
\begin{align}
v_{t}^{*}= & \ v_{1,t},\\
v_{n}^{*}= & \ \alpha v_{1,n}+(1-\alpha)v_{2,n},\label{L201428091}
\end{align}
where we used the expansion of our vectors in normal and tangential
components of $R(\phi)$: 
\[
\vec{v}^{*}=v_{n}^{*}\vec{R}(\phi)+v_{t}^{*}\vec{R}_{\perp}(\phi) .
\]
Similarly for the vectors $\vec{v}_{1},\vec{v}_{2}$ accordingly to
formula (\ref{beforeCollDec}). It is clear that, for any given $\vec{R}(\phi)$,
vector $\vec{v}_{2}$ can be chosen so that the equality (\ref{L201428091})
holds. Define vector $\vec{R}(\phi)$ as follows: 
\[
\vec{R}(\phi)=\frac{1}{|\vec{v}_{1}-\vec{v}^{*}|}(\vec{v}_{1}-\vec{v}^{*}).
\]
It is clear that 
\[
(\vec{v}_{1}-\vec{v}^{*},\vec{R}_{\perp}(\phi))=0,
\]
where $(,)$ is the standard euclidean scalar product on $\mathbb{{\bf R}}^{2}$.
Then 
\[
v_{t}^{*}=v_{1,t}.
\]
and we have proved that point $2$ holds.

Now we shall prove that point $3$ of Condition 2 on $J$ holds. To do this we
expand momentum $p=(p_{1},p_{2})\in\mathbb{{\bf R}}^{2}$ and the
velocity of external particle $\vec{v}$ in normal and tangential
components of the vector $\vec{R}(\phi)$: 
\[
p=p_{n}\vec{R}(\phi)+p_{t}\vec{R}_{\perp}(\phi),\quad\vec{v}=v_{n}\vec{R}(\phi)+v_{t}\vec{R}_{\perp}(\phi).
\]
Then as before we get the formula: 
\[
E_{\xi}|J(\xi;p)|^{2}=p_{t}^{2}+E_{\xi}(\alpha p_{n}+M(1-\alpha)v_{n})^{2}.
\]
Simple calculation gives: 
\[
E_{\xi}|J(\xi;p)|^{2}=|p|^{2}-(1-\alpha^{2})E_{\xi}\{p_{n}^{2}\}\! +\! 2M\alpha(1-\alpha)E_{\xi}\{p_{n}v_{n}\}+M^{2}(1-\alpha)^{2}E_{\xi}\{v_{n}^{2}\} .
\]
We want now to obtain lower estimate of $E_{\xi}\{p_{n}^{2}\}$. We
have: 
\[
E_{\phi}p_{n}^{2}=E_{\phi}(p_{1}\cos\phi+p_{2}\sin\phi)^{2}=f(p_{1},p_{2}),
\]
that defines the quadratic function $f$ on $\mathbb{{\bf R}}^{2}$.
It is clear that $f$ is a non-negatively defined quadratic form with
matrix 
\[
F=\left(\begin{array}{cc}
E\{\cos^{2}\phi\} & E\{\cos\phi\sin\phi\}\\
E\{\cos\phi\sin\phi\} & E\{\sin^{2}\phi\}
\end{array}\right).
\]
Let us prove that $F$ is non-degenerate. In fact, the determinant
of the matrix $F$ is equal to: 
\[
  \det(F)=E\{\cos^{2}\phi\}E\{\sin^{2}\phi\}-(E\{\cos\phi\sin\phi\})^{2}=
\]
\[
  =\frac{1}{4}(1+E\{\cos2\phi\})(1-E\{\cos2\phi\})-\frac{1}{4}(E\sin2\phi)^{2}=
\]
\begin{align*}
=\frac{1}{4}\left(1-((E\cos2\phi)^{2}+(E\sin2\phi)^{2})\right) .
\end{align*}
The distribution $\phi$ was assumed to have density, that is why,
\[
(E\cos2\phi)^{2}+(E\sin2\phi)^{2})<E(\cos^{2}2\phi+\sin^{2}2\phi)=1.
\]
Then 
\[
\det(F)\ne0
\]
and there exists a number $\lambda>0$ such that for all $p=(p_{1},p_{2})\in\mathbb{R}^{2}$
the inequality 
\[
E_{\xi}\{p_{n}^{2}\}=E_{\phi}(p_{1}\cos\phi+p_{2}\sin\phi)^{2}=f(p_{1},p_{2})\geqslant\lambda|p|
\]
holds. It follows that we have the following inequality for transformation
$J$: 
\[
E_{\xi}|J(\xi;p)|^{2}\leqslant\beta|p|^{2}+2M\alpha(1-\alpha)E_{\xi}\{p_{n}v_{n}\}+M^{2}(1-\alpha)^{2}E_{\xi}\{v_{n}^{2}\},
\]
where $\beta=1-\lambda(1-\alpha^{2})<1$. It remains to note that
$E_{\xi}\{p_{n}v_{n}\}$ is linear function of $p$, that proves point 
3.
\end{proof}

\subsection{Proof of Theorem \ref{chainConvThCommonX}.}

Note first that the third proposition follows from the first and second
assertions. Note also that it is the law of large numbers for Markov
chains (see \cite{Revuz}, p.\thinspace 140 and \cite{Skorohod} p.\thinspace 209).

We will use the definitions and the propositions from the book \cite{MT}.
Strong irreducibility property and weak Feller property imply that
the chain $\psi_{k}$ is $\mu$-irreducible, aperiodic (due to theorem
5.4.4, p.\thinspace 113) and satisfies the T-property (thanks to theorem 6.0.1,
p.\thinspace 124). Thus from the theorem 6.0.1 (p.\thinspace 124) it follows that any
compact subset of $X$ is a `petite' set.

Another short proof of the fact that any compact subset is a 'small'
set (property `small' is stronger than `petite') in our case one can
find in \cite{LM_6}. Lyapunov (drift) condition and theorem 9.1.8
(p.\thinspace 206) imply Harris recurrence property of $\psi_{k}$. For the
final proof it remains to use theorem 13.0.1 (p.\thinspace 313) from \cite{MT}.
In fact, we have proved Theorem \ref{thConvChain}.

Note that irreducibility, aperiodicity and T-property for the embedded
chain (\ref{initChain}) one could easily deduce from propositions
for the Nonlinear State Space model (NSS) (see \cite{MT}, p.\thinspace 146)
and our proofs of strong irreducibility and weak Feller property.

\subsection{Dissipative subspace}

Here we will prove Proposition \ref{prop_dissip} together with the
following Proposition containing more detailed information concerning
the dissipating space $L_{-}$ even in more general setting. Let us
remind the definition of $L_{-}$. For simplicity we write $N$ instead
of $dN$. Then fix any subset $\Lambda'\in\{1,\ldots ,N\}$ and define
$L_{-}=L_{-}(\Lambda')$ as the subspace 
\[
L_{-}(\Lambda')=\{(q,p):q,p\in l_{V}\}
\]
where 
\[
l_{V}=l_{V}(\Lambda')=<\{V^{k}e_{n}:n\in\Lambda',k=0,1,2,\ldots \}>,
\]
compare this definition with that in \cite{LM_4}.

$L_{0}=L_{0}(\Lambda')$ is defined as the orthogonal complement $L_{0}=L_{-}^{\perp}$
(with respect to scalar product $(,)_{2}$) to $L_{-}$.

Denote $\mathbf{H}_{ind}$ the subset of $\mathbf{H}$ with the spectrum
of $V$ such that $\omega_{1},\ldots ,\omega_{N}$ are rationally independent.

\begin{proposition}\label{prop61}
  $\phantom{a}$
\begin{enumerate}
\item The set $\mathbf{H}^{+}$ is open, everywhere dense in $\mathbf{H}$
and 
$
\lambda(\mathbf{H\smallsetminus}\mathbf{H}^{+})=0.
$ 
\item The complement of $\,\mathbf{H^{++}=\mathbf{H^{+}\cap}\mathbf{H_{ind}}}$
to $\mathbf{H}$ has Lebesgue measure zero. 
\item For any $n=1,2,\ldots,N$ denote $g_{n}=(0,e_{n})^{T}$, where $\{e_{n}$\}
is the standard basis in $\mathbb{{\bf R}}^{N}$. Then $L_{-}$ is
invariant w.r.t.\ $A$. Moreover, $L_{-}$ can be presented as follows:
\begin{equation}
L_{-}=\ \langle\{A^{k}g_{n}:\ k=0,1,\ldots;n\in\Lambda'\}\rangle,\label{L2014061421}
\end{equation}
\begin{equation}
L_{-}=\ \langle\{(A^{*})^{k}g_{n}:\ k=0,1,\ldots;n\in\Lambda'\}\rangle\label{L2014061422}
\end{equation}
where $A^{*}$ is the adjoint operator to $A$ in the scalar product
$(,)_{2}$, and angle brackets $\langle\ \rangle$ is the linear span
of the corresponding set of vectors.
\item Let $\psi(t)=(q(t),p(t))^{T},$ $q(t)=(q_{1}(t),\ldots,q_{N}(t))^{T},$ $p(t)=(p_{1}(t),\ldots,$ $p_{N}(t))^{T}$
denote the solution of the equation (\ref{LHamVec}) with initial
condition $\psi(0)=\psi$. Then $\psi\in L_{0}$ if and only if for
any $n\in\Lambda'$ and all $t\geqslant0$ the equality $p_{n}(t)=0$
holds. Moreover, the subspace $L_{0}$ is invariant with respect to
$A$. 
\end{enumerate}
\end{proposition}

More physically: the initial conditions from $L_{0}$ are the conditions
where momentum $p_{k}(t)=0$ for all $n\in\Lambda'$ and $t\geq0$.

\paragraph{Proof of 1 and 2}

The fact that $\mathbf{H}^{+}$ is open is evident. It is sufficient
to prove that $\lambda(\mathbf{H}\smallsetminus\mathbf{H}^{+})\! =\! 0$
and $\lambda(\mathbf{H}\smallsetminus\mathbf{H}_{ind})=0$. The first
equality $\lambda(\mathbf{H}\smallsetminus\mathbf{H}^{+})\! =\! 0$ follows
from the fact that $\mathbf{H}\smallsetminus\mathbf{H}^{+}$ is an
algebraic manifold (\hspace{-0.1pt}\cite{LM_4}, lemma 2.3). Let us prove the second
equality. Without loss of generality we can assume that $d=1$. For
the array of numbers $\alpha=(\alpha_{1},\ldots,\alpha_{N})\in\mathbb{R}^{N}$
we define the subset $D_{\alpha}\subset\mathbf{H}$ of the diagonal
matrices with positive diagonal elements $x_{1},\ldots,x_{N}$ such
that 
\[
\sum_{k=1}^{N}\alpha_{k}\sqrt{x}_{k}=0.
\]
Consider the subset $\mathbf{H}_{\alpha}\subset\mathbf{H}$ of matrices
with eigenvalues $\mu_{1},\ldots,\mu_{N}$ such that 
\[
\sum_{k=1}^{N}\alpha_{k}\sqrt{\mu_{k}}=0.
\]
It is easy to see that 
\[
\mathbf{H}_{\alpha}=\{C^{-1}DC:\ D\in D_{\alpha},\ C\in O(N)\},
\]
where $O(N)$ is a set of the orthogonal matrices. In other words,
$\mathbf{H}_{\alpha}$ is an orbit of the $D_{\alpha}$ under the
action of $O(N)$ by the conjugation. So the dimension of $\mathbf{H}_{\alpha}$
equals $N-1+N(N-1)/2<N+N(N-1)/2=N(N+1)/2=\mathrm{dim}(\mathbf{H})$
and hence $\lambda(\mathbf{H}_{\alpha})=0$. We have 
\[
\mathbf{H}_{ind}=\cup_{\alpha\in\mathbb{Q}^{N}}\mathbf{H}_{\alpha}.
\]
And consequently $\lambda(\mathbf{H\smallsetminus}\mathbf{H}_{ind})=0$.
Then Proposition \ref{prop_dissip} is also proved.

\textbf{Proof of 3.} Invariance follows from (\ref{L2014061421}).
Since
\[
A^{2}=-\left(\begin{array}{cc}
V & 0\\
0 & V
\end{array}\right) ,
\]
we get for any $k=0,1,\ldots$ and any $n=1,\ldots,N$ 
\[
A^{2k}g_{n}=(-1)^{k}\left(\begin{array}{c}
0\\
V^{k}e_{n}
\end{array}\right)\quad A^{2k+1}g_{n}=(-1)^{k}\left(\begin{array}{c}
V^{k}e_{n}\\
0
\end{array}\right).
\]
The (\ref{L2014061421}) follows from this. Moreover we have: 
\[
A^{*}=\left(\begin{array}{cc}
0 & -V\\
E & 0
\end{array}\right),\quad(A^{*})^{2}=-\left(\begin{array}{cc}
V & 0\\
0 & V
\end{array}\right)=A^{2}.
\]
Hence for any $k=0,1,\ldots$ and any $n=1,\ldots,N$ we have: 
\[
(A^{*})^{2k}g_{n}=A^{2k}g_{n}\quad(A^{*})^{2k+1}g_{n}=(-1)^{k}\left(\begin{array}{c}
V^{k+1}e_{n}\\
0
\end{array}\right).
\]
Thus we proved that 
\[
\{(A^{*})^{k}g_{n}:\ k=0,1,\ldots;n\in\Lambda'\}=\Bigl\{\Bigl(\begin{array}{c}
q\\
p
\end{array}\Bigr):\ q\in V(l_{V})\ p\in l_{V}\Bigr\}.
\]
Then by Hamilton\tire Cayley theorem for any $n=1,\ldots,N$ we have 
\[
e_{n}\in\langle Ve_{n},\ldots,V^{N}e_{n}\rangle .
\]
 It follows that $V(l_{V})=l_{V}$, that proves (\ref{L2014061422}).

\textbf{Proof of 4.}
For any $\psi\in L$ we have 
\[
e^{tA}\psi=\sum_{k=0}^{\infty}\frac{t^{k}}{k!}A^{k}\psi.
\]
If $n\in\Lambda'$ then for the corresponding momenta: 
\begin{equation}
p_{n}(t)=(e^{tA}\psi,g_{n})_{2}=\sum_{k=0}^{\infty}\frac{t^{k}}{k!}(A^{k}\psi,g_{n})_{2}=\sum_{k=0}^{\infty}\frac{t^{k}}{k!}(\psi,(A^{*})^{k}g_{n})_{2}.\label{20140614321}
\end{equation}
If $\psi\in L_{0}$, then due to Lemma \ref{zeroSet-1} all coefficients
in the latter expansion equal zero, and then $p_{n}(t)=0$ for any
$t\geqslant0$ and any $n\in\Lambda'$. Vice-versa, assume that $p_{n}(t)=0$
for all $t\geqslant0$ and all $n\in\Lambda'$. From expansion (\ref{20140614321})
it follows that $(\psi,(A^{*})^{k}g_{n})_{2}=0$ for all $k=0,1,\ldots$
and all $n\in\Lambda'$. Using Lemma \ref{zeroSet-1} we get that
$\psi\in L_{0}$. Thus, we have proven the first assertion of the
Lemma \ref{-L_gaussian}. Invariance of $L_{0}$ follows because $L_{-}$
is invariant with respect to $A^{*}$ (this is direct consequence
of the formula (\ref{L2014061422})).

Now we want to formulate some results concerning the dimension of
the dissipative subspace in terms of the spectrum of $V$.

\paragraph{The dimension of the dissipative subspace}

Denote $l=\mathbb{{\bf R}}^{N}$. Let $\lambda$ be an eigenvalue
of $V$. Then the subspace 
\[
l(\lambda)=\{u\in l:\ Vu=\lambda u\}
\]
we call the corresponding eigensubspace. Let $\sigma(V)$ be the spectrum
of $V$. As $V$ is symmetric then $l(\lambda)\perp l(\lambda')$
for $\lambda\ne\lambda'$. The space $l$ can be presented as the
direct sum of such eigensubspaces: 
\begin{equation}
l=\bigoplus_{\lambda\in\sigma(V)}l(\lambda). \label{05081310}
\end{equation}
Consider the subspace 
\[
l_{\Lambda'}=\langle\{e_{i}\}_{i\in\Lambda'}\rangle\subset l_{V} .
\]
For $\lambda\in\sigma(V)$ denote $l_{\Lambda'}(\lambda)$ the orthogonal
projection of $l_{\Lambda'}$ onto $l(\lambda)$.

\begin{theorem} \label{0509133} The following expansion holds: 
\[
l_{V}=\bigoplus_{\lambda\in\sigma(V)}l_{\Lambda'}(\lambda).
\]
\end{theorem}

\begin{proof} As $l_{V}$ is invariant w.r.t.\ to $V$, there exists the presentation
of $l_{V}$ as the direct sum of eigensubspaces for the restriction
of $V$ onto $l_{V}$: 
\[
l_{V}=\bigoplus_{\lambda\in\sigma(V)}l_{V}(\lambda),
\]
where $l_{V}(\lambda)=\{u\in l_{V}:\ Vu=\lambda u\}$. Let us prove
that 
\[
l_{V}(\lambda)=l_{\Lambda'}(\lambda)
\]
for all $\lambda\in\sigma(V)$. From the orthogonality of eigensubspaces
it follows that $l_{V}(\lambda)$ is the orthogonal projection of
$l_{V}$ on $l(\lambda)$. Thus, $l_{\Lambda'}(\lambda)\subset l_{V}(\lambda)$.
Let us show that the reverse assertion also holds. Following (\ref{05081310})
we can represent the vector $e_{i},\ i\in\Lambda'$ as: 
\[
e_{i}=\sum_{\lambda\in\sigma(V)}u_{\lambda},
\]
for some $u_{\lambda}\in l(\lambda),\ \lambda\in\sigma(V)$. From
the definition of $l_{\Lambda'}(\lambda)$ it follows that $u_{\lambda}\in l_{\Lambda'}(\lambda)$
for all $\lambda\in\sigma(V)$. Then for any $j=0,1,\ldots$ we get:
\begin{equation}
V^{j}e_{i}=\sum_{\lambda\in\sigma(V)}\lambda^{j}u_{\lambda} . \label{0509131}
\end{equation}
From (\ref{0509131}) it follows that the projection of $V^{j}e_{i}$
on $l(\lambda)$ equals $\lambda^{j}u_{\lambda}$ and thus it belongs
to $l_{\Lambda'}(\lambda)$. Then the projection of $l_{V}$ on $l(\lambda)$,
equals $l_{V}(\lambda)$ and thus belongs to $l_{\Lambda'}(\lambda)$.
The theorem is thus proved.
\end{proof}

From this theorem we have the following corollary.

\begin{lemma} If $L_{-}=L$, then the multiplicity of any eigenvalue
of $V$ cannot exceed $|\Lambda'|$, where $|\Lambda'|$ is the number
of boundary indices $\Lambda'$.
\end{lemma}

\begin{proof} From the definition of $l_{\Lambda'}(\lambda)$ it follows
that for all $\lambda\in\sigma(V)$ the following inequality holds
\[
\dim(l_{\Lambda'}(\lambda))\leqslant\dim(l_{\Lambda'})\leqslant|\Lambda'| .
\]
On the other hand, the condition $L_{-}=L$ is equivalent to that
$l_{V}=l$. This gives the equality $l_{\Lambda'}(\lambda)=l(\lambda)$
for all $\lambda\in\sigma(V)$.
\end{proof}

In particular for $|\Lambda'|=1$ the necessary (but, as we will see
below, not sufficient) condition of completeness of the dissipative
subspace is that the spectrum of $V$ is simple.

\begin{lemma} \label{0509132} Let the spectrum of $V$ be simple
and let $\{v_{1},\ldots,v_{N}\}$ be the basis of the space $l$,
where $v_{k}$ are the eigenvectors of $V$. Then 
\[
\dim(L_{0})=2\#\{k\in\{1,\ldots,N\}:\ v_{k}\in\ l_{\Lambda'}^{\perp}\}.
\]
\end{lemma}

Otherwise speaking Lemma \ref{0509132} says that in case of simple
spectrum of $V$ the dimension of $L_{0}$ is equal to the double
number of the eigenvectors of $V$, having all ``boundary'' coordinates
$i\in\Lambda'$ zero.

\begin{proof} By Lemma \ref{0509133} $l_{V}$ is spanned by those vectors
$v_{k}$, for which $(v_{k},e_{i})_{1}\ne0$ at least for one index
$i\in\Lambda'$. It follows that the complement to $l_{V}$ coincides
with the linear span of the vectors $v_{k}$ for which $(v_{k},e_{i})_{1}=0$
for all $i\in\Lambda'$. The lemma is proven.
\end{proof}

\end{document}